\documentclass{aamas2013}
\usepackage{amssymb,amsmath,amsfonts}
\usepackage{times}

\usepackage{wrapfig}
\usepackage{algorithm}
\usepackage{algorithmic}

\renewcommand{\P}{{\mathrm P}}

\newtheorem{theorem}{Theorem}[section]
\newtheorem{corollary}[theorem]{Corollary}

\newtheorem{lemma}[theorem]{Lemma}

\newtheorem{definition}[theorem]{Definition}

\title{Complex-Demand Knapsack Problems and \\Incentives in AC Power Systems}
\numberofauthors{2}
\author{
\alignauthor
Lan Yu\\
			 \affaddr{Division of Mathematical Sciences}\\
			 \affaddr{School of Physical and Mathematical Sciences}\\
       \affaddr{Nanyang Technological University, Singapore}\\
       \email{yula0001@ntu.edu.sg}
\alignauthor Chi-Kin Chau\\
			 \affaddr{Computing and Information Science}\\
       \affaddr{Masdar Institute of Science and Technology}\\
       \affaddr{Abu Dhabi,UAE}\\
       \email{ckchau@masdar.ac.ae}
}

\begin{document}

\maketitle


\begin{abstract}
We consider AC electrical systems where each electrical device has a power demand expressed as a complex number, and there is a limit on the magnitude of total power supply.  Motivated by this scenario, we introduce the {\em complex-demand knapsack problem} ({\sc C-KP}), a new variation of the traditional knapsack problem, where each item is associated with a demand as a complex number, rather than a real number often interpreted as weight or size of the item.  While keeping the same goal as to maximize the sum of values of the selected items, we put the capacity limit on the magnitude of the sum of satisfied demands.

For {\sc C-KP}, we prove its inapproximability by FPTAS (unless P = NP), as well as presenting a $(1/2-\epsilon)$-approximation algorithm.  Furthermore, we investigate the selfish multi-agent setting where each agent is in charge of one item, and an agent may misreport the demand and value of his item for his own interest.  We show a simple way to adapt our approximation algorithm to be monotone, which is sufficient for the existence of {\em incentive compatible} payments such that no agent has an incentive to misreport.  Our results shed insight on the design of multi-agent systems for smart grid.
\end{abstract}



\section{Introduction}\label{sec:introduction}
\noindent
Most studies of power allocation only consider devices without minimum power requirements; we focus on those with such requirements, such as electric vehicles (EVs) charging, which will not produce any value unless it is charged enough to travel a threshold distance.  Gerding et al. \cite{gerding2011online} studies online electric vehicle charging by expressing power demands as real numbers.  However, in alternating current (AC) electrical systems, alternating power is provided.  In this paper, we study electrical devices with a power demand expressed as a complex number $d = d^{\rm R} + {\bf i} d^{\rm I}$.  Although $d^{\rm I}=0$ for purely resistive appliances; devices with capacitive or inductive components have non-zero imaginary part $d^{\rm I}$ \cite{GS94power}.  

In power allocation, due to the constraint of power generation, there is a limit $C$ on the magnitude of the total power supply, i.e., the magnitude of the sum of satisfied demands should not exceed $C$.  Since only a limited number of devices can be served and different devices produce different values when they receive enough power to work,   
there arises a natural allocation problem: 
we want to select a subset of devices to provide power subject to the power limit constraint such that the total value produced is maximized.

Moreover, in a multi-agent setting (e.g., in future smart grid, where intelligent devices are automatically controlled by agents), the demand and value of each device are private knowledge of an individual agent.  The power allocation algorithm collects the input information from each agent, and based on that, computes which subset of demands to satisfy.  Depending on the (publicly known) algorithm, each selfish agent may misreport his demand or value to the algorithm in order to get selected.  

Naturally, to guarantee a good realization of our optimization goal (here, maximizing {\em social welfare}, the total value of selected items), we would like to design the algorithm in a way that incentivizes all agents to report their true information.  
This falls into the study of {\em Algorithmic mechanism design} \cite{N07book,NR01alg}, a burgeoning research area that deals with designing algorithms (called {\em mechanisms} here) for settings where inputs are controlled by selfish agents.  
Each agent is modeled to strategize so as to maximize his {\em utility}, a quantity that indicates his overall benefit.
A mechanism is {\em incentive compatible}, or simply {\em truthful}, if no agent has an incentive to misreport.
A general approach in mechanism design is to enforce payment on each agent to adjust his utility so that truth-telling always maximizes his utility.   

Now formally, we have the following mechanism design problem: we have a set $K$ of distinct agents where each agent $k\in K$ owns an item with a positive value $v_k$ and a complex-valued demand $d_k=d^{\rm R}_k + {\bf i} d^{\rm I}_k$.  Given capacity $C>0$, our task is to choose a subset $S\subseteq K$ of agents to satisfy their demands and assign each agent $k$ a nonnegative payment $p_k$.  The goal is to elicit true inputs and maximize the total value of selected items subject to the constraint that $|\sum_{k\in S} d_k|\leq C$.       	
Here we limit our attention to the case where $d^{\rm R}_k, d^{\rm I}_k\geq 0$ for all $k$. This assumption is reasonable, since, although demands do not necessarily lie in the first quadrant of the complex plane, they are recommended\footnote{
NEC NFPA 70-2005 (a standard for electrical systems and appliances) suggests that high-consumption appliances should conform to restricted power factor, which implies $d^{\rm R}_k \geq |d^{\rm I}_k|$.} to stay within the region $d^{\rm R} \geq |d^{\rm I}|$, which can be obtained by rotating the first quadrant by $\pi/4$.
 
Usually, the design of a truthful mechanism is composed of two steps: first, we solve the pure algorithmic problem; second, we identify certain condition that guarantees the existence of incentive compatible payments and make it satisfied by our algorithm.  

We follow this path for our problem.  Our algorithmic problem is a winner determination optimization problem; we call it the {\em complex-demand knapsack problem} ({\sc C-KP}), as it turns out to be an interesting new variation of the traditional knapsack problem \cite{KPP10book}.  In the original {\em one-dimensional knapsack problem} ({\sc 1-KP}), the demand of an item is simply a nonnegative real number, often interpreted as the weight or size of the item.  The "knapsack", with fixed real-valued capacity to hold the items, represents the limited resource.  The multi-dimensional generalization, the {\em $m$-dimensional knapsack problem} ({\sc $m$-KP}), captures the settings where there are $m$ independent resource constraints on the $m$ dimensions (independent features) of the demands.  1-KP can model the power allocation in direct current (DC) electrical systems, where power demands can be expressed as real numbers, but fails for AC systems, where each demand is a two-dimensional vector.  Our problem is also different from {\sc 2-KP} since our capacity constraint is a quadratic one (on the magnitude of the total satisfiable demand), rather than two independent linear constraints in {\sc 2-KP}.  Moreover, it is natural to modify our problem by including these two linear constraints (on the real and imaginary part of the total satisfiable demand respectively), and thus introduce the {\em generalized complex-demand knapsack problem} ({\sc GC-KP}).  In fact, many power generators do have all three constraints of GC-KP.     

It is well-known that 1-KP is NP-hard, and our complex-demand variations include it as a special case when we set all $d^{\rm I}_k=0$.  
Hence we are interested in good polynomial-time approximation algorithms.  In this work, we present an algorithm with constant $1/2-\epsilon$ approximation ratio for both {\sc C-KP} and {\sc GC-KP}, and show the inapproximability of {\sc C-KP} by FPTAS (unless P = NP), based on its connections to well-studied {\sc 1-KP}, {\sc 2-KP} and {\sc 3-KP}.  
There is still a gap to close, and we conjecture that {\sc C-KP} admits a PTAS.


As to the incentive part, the difficulty lies in the following: VCG mechanisms \cite{N07book} are both social welfare maximizing and truthful; however, they become computationally infeasible when computing optimal social welfare is computationally hard, as in our setting.  Worse still, using algorithms approximating maximum social welfare may not preserve truthfulness.  
To obtain truthful and efficient mechanisms with a good approximation ratio, a leading approach is through "monotonization": First prove that a certain notion of monotonicity suffices for the existence of incentive compatible payments and then design or adapt an existing algorithm to be monotone.  This has been successfully applied to problem settings with {\em single-parameter} \cite{archer2001truthful} and {\em single-minded} agents \cite{LOS99mono}, with efficiently computable payments specified; in fact, for the former, monotonicity is necessary as well, which justifies the necessity of monotonization.  An additional nice property of the specified payments is that they guarantee nonnegative utilities for all agents, which, in mechanism design, is an important desired property called {\em individual rationality} ensuring voluntary participation of the agents. 

For the knapsack problem, if both demand and value of an item are private information, which is the case we investigate here, we do not have single-parameter agents. However, all variations we consider are special cases of single-minded agents, each has a single object $d_k$ in mind, gets value $v_k$ if he is assigned an object no worse than $d_k$ and 0 otherwise.  For example, in our power system setting, the power demand $d_k$ is the single object the $k$th agent desires, and the value $v_k$ is produced as long as the power he receives is $\geq d_k$ (according to comparisons between complex numbers).  The monotonicity property for single-minded agents looks natural and reasonable: If an agent is selected with certain demand and value, he should remain selected with a lower demand and a higher value, while the inputs of other agents are fixed.  Although this property easily holds for exact optimization, 
it may not hold for approximation algorithms.   
For {\sc C-KP}, we succeed in monotonizing our constant approximation algorithm, based on an existing monotone FPTAS for {\sc 1-KP} in \cite{BKV05KS}, and thus achieving incentive compatibility.


\smallskip

\noindent{\bf Related Work\ }
The knapsack problem has many variations with respect to divisibility of items, copies of items, dimensions of constraints, etc \cite{KPP10book}.
In this work, we restrict our attention to the NP-hard {\em one-dimensional knapsack problem} ({\sc 1-KP}) where each indivisible item has only one single copy, and its multi-dimensional generalization, the {\em $m$-dimensional knapsack problem} ({\sc $m$-KP}).

For {\sc 1-KP}, there is a pseudo-polynomial time algorithm using dynamic programming achieving exact optimization when all item values are integers.  There is a simple {\em fully polynomial-time approximation scheme} (FPTAS), which scales and rounds the item values and then applies the pseudo-polynomial time algorithm on small integer values \cite{KPP10book}.  However, this FPTAS is not monotone, since the scale factor involves the maximum item value.  Briest et al. \cite{BKV05KS} monotonized it, by performing the same procedure with a series of different scaling factors irrelevant to item values and taking the best solution out of them.  Hence {\sc 1-KP} admits an incentive compatible FPTAS.       

As to {\sc $m$-KP} with $m\geq 2$, there is a {\em polynomial-time approximation scheme} (PTAS) by Frieze and Clarke \cite{FC84alg} based on the integer programming formulation, but it is not evident to see whether it is monotone.  On the other hand, {\sc 2-KP} is already inapproximable by FPTAS unless P = NP, by a reduction from {\sc equipartition} \cite{KPP10book}.  In fact, there is no {\em efficient polynomial-time approximation scheme} (EPTAS) for {\sc 2-KP} unless W[1] = FPT (See \cite{kulik2010there}).
 
\smallskip

\noindent{\bf Our Results\ }
We initiate the study of the {\em complex-demand knapsack problem} ({\sc C-KP}) and its hybrid with {\sc 2-KP}, the {\em generalized complex-demand knapsack problem} ({\sc GC-KP}).

In Section \ref{sec:algCKP}, we present an approximation algorithm for C-KP, which projects all demand vectors onto the $\pi/4$ line and uses an approximation algorithm for 1-KP as a subroutine.  Since 1-KP admits an FPTAS, we achieve approximation ratio $1/2-\epsilon$ for any $\epsilon>0$, with running time polynomial in $1/\epsilon$ and the size of the input.  
Moreover, the 
algorithm can be monotonized, 
as shown in Section \ref{sec:monalg}, due to the existence of the monotone FPTAS for 1-KP.  

On the other hand, in Section \ref{sec:inappr}, we complete our study of C-KP by providing an inapproximability result.  We prove that there is no FPTAS for C-KP unless P = NP, through a modification of the reduction from {\sc equipartition} for 2-KP.  

Finally, for GC-KP, the inapproximability result is inherited since it includes C-KP as a special case.  We also come up with an approximation algorithm by applying the same idea as for C-KP, but we have to use a PTAS for 3-KP as a subroutine (Section \ref{sec:algGCKP}). Again we achieve approximation ratio $1/2-\epsilon$ for any $\epsilon>0$, but the running time is only guaranteed to be polynomial in the size of the input.   
Regarding monotonization, a similar trick as in Section \ref{sec:monalg} would work for {\sc GC-KP}, if we could find a good monotone approximation algorithm for {\sc 3-KP}.


\section{Preliminaries}\label{sec:prelim}
\noindent
\subsection{The Knapsack Problems}
\noindent
Here we give the integer programming formulation of the knapsack problems discussed in this paper.  The decision of an allocation algorithm is specified by indicator variables $x_k\in \{0,1\}$ for item $k\in K$, which has a simple correspondence to the selected subset of items: $S= \{k\in K: x_k=1\}$.  We will switch back to the subset representation in later sections for convenience of illustration. 

The {\em one-dimensional knapsack problem} ({\sc 1-KP}) is defined as: 
\begin{equation*}
\begin{array}{@{}lc} \hspace{-50pt}
\mbox{({\sc 1-KP})} & \displaystyle \qquad \max \sum_{k \in K} x_k v_k 
\end{array}
\end{equation*}
subject to 
\begin{equation*}
\sum_{k \in K} x_k d_k \le C 
\end{equation*}
where 
\begin{itemize}

\item $K$ is a set of items;

\item $v_k$ is the positive value of item $k$ if its demand is satisfied;

\item $d_k$ is the nonnegative real-valued demand of item $k$;

\item $C$ is the positive real-valued capacity on the total satisfiable demand;

\item $x_k$ indicates whether item $k$ is selected: $x_k = 1$ means that the demand of item $k$ is satisfied, and 0 otherwise.

\end{itemize}

1-KP can be generalized to multi-dimensions.  The {\em $m$-dimensional knapsack problem} ({\sc $m$-KP}) is defined as: 
\begin{equation*}
\begin{array}{@{}lc} \hspace{-50pt}
\mbox{({\sc $m$-KP})} & \displaystyle \qquad \max \sum_{k \in K} x_k v_k 
\end{array}
\end{equation*}
subject to $m$ independent inequalities
\begin{equation*}
\sum_{k \in K} x_k d^j_k \le C^j 
\end{equation*}
for $j=1,\ldots, m$, where 
\begin{itemize}

\item $d^j_k$ is the nonnegative real-valued demand of item $k$ in dimension $j$;

\item $C^j$ is the positive real-valued capacity on the total satisfiable demand in dimension $j$.

\end{itemize}
Each $m$-KP is a linear integer program, and $m$-KP is a special case of $(m+1)$-KP for all $m$.  We are especially interested in 1-KP, 2-KP and 3-KP, whose previous results will be used to achieve ours.  In particular, the {\em two-dimensional knapsack problem} ({\sc 2-KP}) can also be formulated in terms of complex-valued demands:


\begin{equation*}
\begin{array}{@{}lc} \hspace{-50pt}
\mbox{({\sc 2-KP})} & \displaystyle \qquad \max \sum_{k \in K} x_k v_k 
\end{array}
\end{equation*}
subject to 
\begin{equation*}
\sum_{k \in K} x_k d_k^{\rm R} \le C^{\rm R} \mbox{\ and\ } \sum_{k \in K} x_k d_k^{\rm I} \le C^{\rm I} 
\end{equation*}
where 
\begin{itemize}

\item $d^{\rm R}_k, d^{\rm I}_k$ are the nonnegative real part and imaginary part respectively of the complex-valued demand $d_k$ of item $k$;

\item $C^{\rm R}, C^{\rm I}$ are the positive real-valued capacities on the real part and imaginary part respectively of the total satisfiable demand.
\end{itemize}


Our study concerns the capacity constraint on the magnitude of the total satisfiable demand, which is no longer linear.  We formulate the {\em complex-demand knapsack problem} ({\sc C-KP}) as follows: 
\begin{equation*}
\begin{array}{@{}lc} \hspace{-50pt}
\mbox{({\sc C-KP})} & \displaystyle \qquad \max \sum_{k \in K} x_k v_k 
\end{array}
\end{equation*}
subject to  
\begin{equation*}
\Big|\sum_{k \in K} x_k d_k \Big| \le C 
\end{equation*}
where
\begin{itemize}

\item $d_k = d_k^{\rm R} + {\bf i} d_k^{\rm I}$ is the complex-valued demand of item $k$ where $d^{\rm R}_k, d^{\rm I}_k$ are both nonnegative;

\item $C$ is the positive real-valued capacity on the magnitude of the total satisfiable demand. 
\end{itemize}

Combining the constraints of C-KP and 2-KP results in the following {\em generalized complex-valued knapsack problem} ({\sc GC-KP}): 
\begin{equation*}
\begin{array}{@{}lc} \hspace{-50pt}
\mbox{({\sc GC-KP})} & \displaystyle \qquad \max \sum_{k \in K} x_k v_k 
\end{array}
\end{equation*}
subject to  
\begin{equation*}
\Big|\sum_{k \in K} x_k d_k \Big| \le C \mbox{\ and\ } \sum_{k \in K} x_k d_k^{\rm R} \le C^{\rm R} \mbox{\ and\ } \sum_{k \in K} x_k d_k^{\rm I} \le C^{\rm I}.  
\end{equation*}

\subsection{Approximation Algorithm}
\noindent
For knapsack problems, given a solution represented by the selected subset of items $S\subseteq K$, we denote the total value of selected items by $v(S)=\sum_{k\in S}v_k$.  Let $S^\ast$ denote an optimal solution.

For our value maximization objective, an algorithm is called a {\em $\rho$-approximation}, if on each input, the output $S$ of the algorithm satisfies $v(S) \ge \rho \cdot v(S^\ast)$.   
Since the knapsack problems considered in this paper are NP-hard, one looks for polynomial-time algorithms with good approximation ratio $\rho$. 

It is desirable to find constant approximation algorithms with $\rho$ as close to 1 as possible; stronger than that are algorithms whose approximation ratio can be arbitrarily close to 1: 

One such candidate is a {\em polynomial-time approximation scheme} (PTAS), which is a $(1-\epsilon)$-approximation algorithm for any $\epsilon>0$.  The running time of a PTAS is polynomial in the input size for every fixed $\epsilon$, but the exponent of the polynomial might depend on $1/\epsilon$.  One way of addressing this is to define the {\em efficient polynomial-time approximation scheme} (EPTAS), whose running time is the multiplication of a function in $1/\epsilon$ and a polynomial in the input size independent of $\epsilon$.
An even stronger notion is a {\em fully polynomial-time approximation scheme} (FPTAS), which requires the running time to be polynomial in both the input size and $1/\epsilon$.  

In this work, we design constant $1/2-\epsilon$ approximation algorithms for C-KP and GC-KP based on the FPTAS for 1-KP and PTAS for 3-KP respectively. 
\subsection{Incentive Compatibility}\label{subsec:ic}
\noindent
In this subsection, we give a formal model of mechanism design with single-minded agents based on our C-KP problem setting, state the monotonicity condition, and specify the incentive compatible payments under it.  Single-minded agents are first introduced by Lehmann et al. \cite{LOS99mono}, and here we essentially present the model described in \cite{BKV05KS}.  Readers can refer to \cite{N07book, NR01alg} for a formal definition of the general setting of mechanism design.

We are given a set $K$ of agents, where agent $k$ controls item $k$.  The demand and value of item $k$ is agent $k$'s private information, which is called his {\em type}, denoted by $t_k=(d_k, v_k)$.   Each agent $k$ is single-minded: he has the single demand $d_k$ in mind, and enjoys value $v_k$ if and only if his demand is satisfied.  

Here, with selfish behaviors, satisfying the demand of an agent is no longer the same as selecting an agent, since an agent may get selected by reporting a lower demand, but the assignment he receives is only guaranteed to cover his reported demand, which may not be enough for his true demand.  Therefore, we need to modify an outcome $o$ of an allocation algorithm from the indicator variable $x_k\in \{0,1\}$ for each agent $k$ to a specific assignment $o_k$ agent $k$ receives (clearly $o_k=0$ when $x_k=0$).  Let $\mathbb{C}_+$ denote all complex numbers in the first quadrant of the complex plane, we have $o_k\in \mathbb{C}_+$ and $o\in \mathbb{C}_+^{K}$.  

Now we are able to represent the value agent $k$ derives from an outcome $o$ by his {\em valuation function}: $t_k(o)=v_k$ if $d_k\leq o_k$ and 0 otherwise.  Conventionally we abuse the notation and use $t_k:\mathbb{C}_+^{K}\rightarrow \mathbb{R}$ to denote the valuation function associated with type $t_k$.  The comparison $d_k\leq o_k$ interprets the condition that the assignment meets the demand.
For C-KP, it conforms to the partial order between complex numbers: $z_1\leq z_2$ iff $z_1^{\rm R}\leq z_2^{\rm R}$ and $z_1^{\rm I}\leq z_2^{\rm I}$.  It can also be generalized to settings where the outcome set admits a partial order and a minimum element.  As required in the general model of mechanism design, our valuation function only depends on the outcomes, which also justifies the necessity to change our representation of outcomes.  

For ease of notation, we let $t$ denote an input, a list of all agents' types $((d_k,v_k):k\in K)$ and denote the input except that of agent $k$ by $t_{-k}$.  Clearly $t=(t_k,t_{-k})$.\footnote{Unless specified as the true type, $t_k$ may denote any reported type.}

A {\em mechanism} $M=({\cal A},p)$ consists of an allocation algorithm $\cal A$ computing an allocation solution ${\cal A}(t)\in \mathbb{C}_+^{K}$ for each input $t$ and a $|K|$-tuple $p(t)$ for each $t$ where $p_k(t)\in \mathbb{R}$ is the payment enforced on agent $k$.  If $d_k\leq {\cal A}(t)_k$, we say that agent $k$ is selected, i.e., he receives an assignment that meets his input demand.  We represent the set of selected agents as $S(A(t))$. 
Given the mechanism, the {\em utility}, the overall benefit of agent $k$, when his true type is $t_k$, equals his valuation minus the payment: $u_k(t)= t_k({\cal A}(t))-p_k(t)$.

As mentioned in Section \ref{sec:introduction}, given the mechanism, each agent may not report his true type for his own benefit.  Suppose agent $k$ has true type $t_k=(d_k, v_k)$ and reports $t_k'=(d_k', v_k')$.  Here the outcome of the algorithm $\cal A$ is ${\cal A}(t_k',t_{-k})$, but his valuation function remains $t_k$, so he obtains valuation $t_k({\cal A}(t_k',t_{-k}))$, and his utility is $u_k(t_k',t_{-k})= t_k({\cal A}(t_k',t_{-k}))-p_k(t_k',t_{-k})$.  On the other hand, if he reports his true type, his utility is $u_k(t_k,t_{-k})= t_k({\cal A}(t_k,t_{-k}))-p_k(t_k,t_{-k})$.  Each selfish agent intends to maximize his utility, so he will choose to misreport $t_k'$ if it results in higher utility, assuming other agents do not change their input, i.e., $u_k(t_k',t_{-k})> u_k(t_k,t_{-k})$.  Therefore, a mechanism is {\em incentive compatible}, or {\em truthful}, if and only if this can not happen, which is equivalent to saying that, for any agent $k$, any $t_{-k}$ and any true type $t_k$, truth-telling maximizes agent $k$'s utility, i.e., $u_k(t_k,t_{-k})\geq u_k(t_k',t_{-k})$ for any possible $t_k'$.    

A sufficient condition to ensure truthfulness for single-minded agents is {\em monotonicity}, specified as follows in our setting:

\begin{definition}\label{def:mon}
An allocation algorithm ${\cal A}$ is {\em monotone} if $k\in S({\cal A}(t_k,t_{-k}))$ implies $k\in S({\cal A}(t_k',t_{-k}))$ for any $t_k=(d_k,v_k)$ and $t_k'=(d_k',v_k')$ with $v'_k \ge v_k$, ${d}_k' \le d_k$. 
\end{definition}
Intuitively, in a monotone algorithm, if agent $k$ is selected with demand $d_k$ and value $v_k$, he should be also selected when he has smaller demand ${d}_k'$ and larger value $v'_k$.\footnote{Note that in this definition, the specific assignments $o_k$ are irrelevant, so in Section \ref{sec:monalg}, we can stay with our original problem formulation when we argue about the monotonicity of our algorithm.}  
The following theorem states the sufficiency of monotonicity \cite{BKV05KS,LOS99mono}:
\begin{theorem}
\label{thm:IC}
Let ${\cal A}$ be a monotone and exact algorithm for single-minded agents.  Then there exists payment $p^{\cal A}$ such that ${\cal M}_{\cal A}=({\cal A},p^{\cal A})$ is incentive compatible.
\end{theorem}

We call a mechanism {\em exact} if for all inputs $t=((d_k,v_k):k\in K)$ and all agents $k$, $A(t)_k$ is either $d_k$ or $0$, i.e., either the exact demand is satisfied or nothing is assigned.  Without exactness, an agent may benefit from underreporting his demand.  It is not difficult to see that we can always modify a truthful mechanism to be exact.  After all, exactness is a reasonable assumption since it is undesirable to waste resource in our allocation.     

The incentive compatible payment $p^{\cal A}$ is specified as follows:
Given a monotone algorithm ${\cal A}$, if we fix $d_k$ and $t_{-k}$ for agent $k$, then ${\cal A}$ defines a {\em critical value} $\theta_k(d_k, t_{-k})$, such that when $v_k$ is above the critical value, $k$ is selected; and when $v_k$ is below the critical value, $k$ is not selected.  Then we can define a payment function $p^{\cal A}(t)$, where each selected agent pays the critical value: 
\begin{equation*} 
p_k^{\cal A}(t)= 
\left\{
\begin{array}{ll}
\theta_k(d_k,t_{-k}) & \mbox{if agent $k$ is selected\ } \\
0 & \mbox{otherwise\ }  \\
\end{array} 
\right.  
\end{equation*} 

By Theorem~\ref{thm:IC}, if we are able to design a monotone algorithm, we can transform it into a truthful mechanism.  Moreover, the critical value for a given input can be computed in polynomial time by a binary search on interval $[0, v_k]$ for each agent $k$ during which we repeatedly test if $k$ is satisfied by running algorithm ${\cal A}$. Therefore, a monotone polynomial time allocation algorithm ${\cal A}$ implies a polynomial time truthful mechanism.

In addition, the payment function $p^{\cal A}(t)$ guarantees that all agents receive nonnegative utilities.  This property, called {\em individual rationality}, ensures voluntary participation of the agents, thus is also an important desired property in mechanism design.   
 
Therefore, the monotone polynomial time algorithm for C-KP we will present in Section \ref{sec:monalg} implies a polynomial time mechanism that is both individually rational and incentive compatible.  

We need to point out that the mechanism requires the item values $\{v_k\}$ to be integers, because of the binary search needed in the payment computation.  This is a reasonable assumption, since values are usually rounded up to the nearest cent or dollar.  The approximation algorithm in Section \ref{sec:algCKP} does not need this assumption, since the FPTAS for 1-KP rounds the item values.

\section{Approximation Algorithm for \\C-KP}\label{sec:algCKP}
\noindent
We present a polynomial-time $(\frac{1}{2}-\epsilon)$-approximation algorithm for {\sc C-KP}, which relies on a polynomial-time approximation algorithm for {\sc 1-KP} as a subroutine.   
\subsection{Basic Idea} \label{subsec:pic}
\noindent
Graphically, each demand $d_k=d^{\rm R}_k+ {\bf i}d^{\rm I}_k$ of item $k$ is a vector in the first quadrant.  A feasible solution of our problem is a subset of items whose sum of demands lies in region ${\cal D}$, the $1/4$ disk of radius $C$ in the first quadrant. 
As shown in Fig.~\ref{fig:fig1}, ${\cal D}$ is divided by chord $PQ$ into a closed triangle ${\cal D}_1$ and a circular segment ${\cal D}_2={\cal D}-{\cal D}_1$.  The $\frac{\pi}{4}$ line intersects chord $PQ$ at point $R$.  Since we may preprocess the demands and eliminate those whose magnitude exceeds capacity $C$, without loss of generality, we assume all $|d_k|\leq C$.  

\begin{figure}[htb!]
 \centering 
 \includegraphics[scale=0.5]{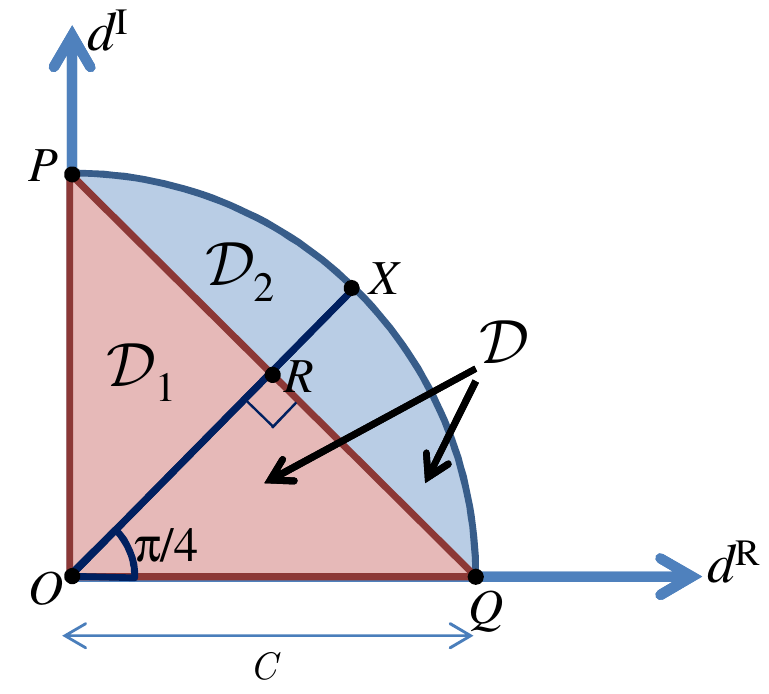} 
 \caption{Graphical picture for C-KP.} 
\label{fig:fig1}
\end{figure}


If we project all demands onto the $\frac{\pi}{4}$ line, i.e., 
\begin{equation*}
\tilde{d}_k \triangleq (d_k^{\rm R}+d_k^{\rm I})/\sqrt{2},
\end{equation*}
we make all demands one-dimensional.  Now a subset of demands has sum $\sum \tilde{d}_k \le  C/\sqrt{2}$ (i.e., the sum vector does not go beyond point $R$ on the $\frac{\pi}{4}$ line) if and only if its original sum vector $\sum d_k$ lies inside the triangle ${\cal D}_1$.  This is because that, the sum of projections, $\sum \tilde{d}_k$, is the projection of $\sum d_k$ on the $\frac{\pi}{4}$ line.  Therefore, the subproblem on feasible region ${\cal D}_1$ can be solved by an approximation algorithm for {\sc 1-KP} with demands changed to $\tilde{d}_k$ and capacity to $C/\sqrt{2}$.

On the other hand, the subproblem on feasible region ${\cal D}_1$ is almost the whole story: First, evidently an optimal solution in ${\cal D}$ can contain at most one demand in ${\cal D}_2$; second, if an optimal solution consists of more than one demand, its sum can be broken into either two separate subsums lying in ${\cal D}_1$, or, the sum of a vector in ${\cal D}_2$ and a subsum in ${\cal D}_1$.  Our algorithm takes the maximum between an approximate solution for the subproblem on feasible region ${\cal D}_1$ and an optimal solution on input demands lying in ${\cal D}_2$.  This only reduces the approximation ratio by at most a factor of 2.


\subsection{Approximation Algorithm}
\noindent 
We let ${\sf Alg}^{\rm a} [(d_k,v_k: k \in K), C]$ be our algorithm for {\sc C-KP}, where $(d_k,v_k: k \in K)$ are the complex-valued demands and values of items and $C$ is the capacity. Moreover, we let ${\sf Alg}^{\rm 1d}[(d_k,v_k: k \in K), C]$ be a polynomial-time approximation algorithm for {\sc 1-KP}, where each demand is real-valued.
We describe our algorithm as follows:

\begin{algorithm}[htb!]
\caption{${\sf Alg}^{\rm a} [(d_k,v_k: k \in K), C]$}
\begin{algorithmic}[1]
\FOR{$k \in K$} 
\STATE Set $\tilde{d}_k =  \frac{d_k^{\rm R}+d_k^{\rm I}}{\sqrt{2}} $
\ENDFOR
\STATE Set $S_1 = {\sf Alg}^{\rm 1d}[(\tilde{d}_k,v_k: k \in K),\frac{C}{\sqrt{2}}]$
\STATE Set $S_2 = \{\displaystyle {\arg\max}_{k \in K : d_k \in {\cal D}_2  } v_k \}$
\STATE Set $S = \arg\max_{S_1, S_2}\{ v(S_1), v(S_2) \}$
\STATE Output $S$
\end{algorithmic}
\end{algorithm}

In ${\sf Alg}^{\rm a}$, we first project all demands onto the $\frac{\pi}{4}$ line, 
and use an approximation algorithm ${\sf Alg}^{\rm 1d}$ for {\sc 1-KP} to compute an allocation (denoted by $S_1$) considering the projected demands and capacity $C/\sqrt{2}$.  Then we look at all demands lying in region ${\cal D}_2$ and choose one with maximum value as solution $S_2$.  Note that $S_2$ only consists of a single item.  Finally, we compare the total value of solutions $S_1$ and $S_2$ and pick the larger one.  All ties are broken arbitrarily.

\subsection{Analysis}
\noindent
It is evident that our algorithm outputs a feasible solution in polynomial time.  For the approximation ratio, our main result is:

\begin{theorem} 
\label{thm:2apx}
If ${\sf Alg}^{\rm 1d}$ is a $\rho$-approximation algorithm for {\sc 1-KP}, then ${\sf Alg}^{\rm a}$ is a $\frac{\rho}{2}$-approximation algorithm for {\sc C-KP}.
\end{theorem}

\begin{corollary}
\label{cor:2apx}
Since {\sc 1-KP} has an FPTAS \cite{BKV05KS, KPP10book}, there is a $(\frac{1}{2}-\epsilon)$-approximation algorithm for {\sc C-KP} that runs in polynomial-time in the size of input and $1/\epsilon$, for any $\epsilon>0$.
\end{corollary}   

Now we prove Theorem \ref{thm:2apx}.
\begin{proof}
Let $S^*$ be an optimal solution to C-KP, for which the feasible region is $\cal D$.  Let $S_1^*$, $S_2^*$ be an optimal solution for the subproblem on feasible region ${\cal D}_1$ and ${\cal D}_2$ respectively.  By our observation in Subsection \ref{subsec:pic}, $S_1^*$ is an optimal solution to 1-KP on projected demands and capacity $C/\sqrt{2}$.  Since ${\sf Alg}^{\rm 1d}$ is a $\rho$-approximation algorithm to {\sc 1-KP}, we have $v(S_1)\geq \rho \cdot v(S_1^*)$.  It is also evident that $v(S_2^*)=v(S_2)$.

Next, we analyze the approximation ratio of ${\sf Alg}^{\rm a}$ in three cases.  Here for a subset $S \subseteq K$, we define 
\begin{equation*}
d(S) \triangleq \sum_{k\in S} d_k=\sum_{k\in S} d_k^{\rm R}+ {\bf i}\sum_{k\in S} d_k^{\rm I} 
\end{equation*}

\noindent {\bf Case (1): ($\rho$-approximation) } We consider an optimal solution $S^\ast$, such that its sum of demands $d({S^\ast})\in {\cal D}_1$.  

This is an easy case where $v({S^\ast})=v({S_1^\ast})$.  We have $v(S) \ge v(S_1)\ge  \rho \cdot v({S_1^\ast})=\rho \cdot v({S^\ast})$.
\vskip 5pt

\noindent {\bf Case (2):  ($\frac{\rho}{1+\rho}$-approximation) } We consider an optimal solution $S^\ast$, such that $d({S^\ast})\in {\cal D}_2$, and there exists an item $j\in S^\ast$ whose demand $d_j\in {\cal D}_2$.  

Let $z \triangleq \sum_{k\in S^\ast\setminus\{j\}} d_k$. Thus, $d({S^\ast})=d_j+z$, i.e., the sum of demands of $S^\ast$ can be written as the sum of a single demand $d_j$ and a subset sum $z$.\footnote{It is possible that $S^{\ast}$ only consists of a single item $j$, in which case our algorithm obviously produces the optimal answer.}  Note that $d_j\in {\cal D}_2$ and $z\in {\cal D}_1$. Otherwise, the projection of $d({S^\ast})=d_j+z$ on the $\frac{\pi}{4}$ line would exceed $2\cdot C/\sqrt{2}>C$.  

Moreover, we have $v({S^\ast\setminus\{j\}}) \le  v({S_1^\ast})$, because $S_1^\ast$ is an optimal solution for feasible region ${\cal D}_1$.  On the other hand, $v_j \le v({S_2})$ since item $j$ with $d_j\in {\cal D}_2$ is a candidate for $S_2$ in our algorithm.  
We obtain:
\begin{equation*}
v({S^\ast})=v_j+v({S^\ast\setminus\{j\}}) \le  v({S_2})+v({S_1^\ast})
\end{equation*}

By the description of our algorithm, the total value of the output solution $v(S)=\max(v({S_1}), v({S_2}))\geq \max(\rho\cdot v({S_1^*}), v({S_2}))= \max(\rho\cdot v({S_1^*}), v({S_2^*}))$.  Now it remains to show that it is further $\geq \frac{\rho}{1+\rho}(v({S_2})+v({S_1^*}))$.

If $\rho\cdot v({S_1^*})\geq v({S_2})$, we have that $v(S)$ is at least 
\begin{equation*}
\rho\cdot v({S_1^*})= \frac{\rho}{1+\rho}(\rho\cdot v({S_1^*})+v({S_1^*}))
\geq \frac{\rho}{1+\rho}( v({S_2})+v({S_1^*}));
\end{equation*}
otherwise, $v(S)$ is at least  
\begin{equation*}
v({S_2})= \frac{\rho}{1+\rho}(v({S_2})+\frac{1}{\rho}v({S_2}))\geq \frac{\rho}{1+\rho}( v({S_2})+v({S_1^*})).
\end{equation*}


\noindent {\bf Case (3): ($\frac{\rho}{2}$-approximation) } We consider an optimal solution $S^\ast$, such that $d({S^\ast})\in {\cal D}_2$, and $d_k\in {\cal D}_1$ for every item $k\in S^\ast$.

First, we let $\tilde{d}(S) \triangleq \sum_{k\in S} \tilde{d}_k$.  The condition on $S^{\ast}$ is equivalent to the following condition on 
projected demands on the $\frac{\pi}{4}$ line: $C/\sqrt{2}<\tilde{d}({S^\ast}) \leq C$, and $\tilde{d}_k \le  C/\sqrt{2}$ for every item $k\in S^\ast$.

We use Lemma \ref{lem:subsetsumA} to show that $d({S^\ast})\in {\cal D}_2$ can be written as the sum of two demand subset sums in ${\cal D}_1$.  Lemma \ref{lem:subsetsumA} is essentially an equivalent statement of this on the projected demands, and will be proved later in this subsection.     

\begin{lemma} \label{lem:subsetsumA}
For a set of $n$ positive real numbers $a_1, ..., a_n$ satisfying $\sum_{i = 1}^n a_i \le C$,  $a_i \le C'$ for all $i$ and $C'\ge C/\sqrt{2}$, there exists a subset $T \subseteq \{1,...,n\}$ such that
\begin{equation*}
\sum_{i \in T} a_i \le C' \mbox{\ \ and \ \ } \sum_{i \in \{1,...,n\} \backslash T} a_i \le C'.
\end{equation*}
\end{lemma}

By Lemma \ref{lem:subsetsumA}, we have $\tilde{d}(T)$ and $\tilde{d}({S^{\ast}\setminus T}) \le  C/\sqrt{2}$ for some subset $T \subseteq S^\ast$. That is, $d(T)  \in {\cal D}_1$ and $d({S^\ast\setminus T}) \in {\cal D}_1$.

Thus, $v(T) \le  v({S_1^\ast})$ and $v({S^\ast\setminus T}) \le  v({S_1^\ast})$.  Moreover, since $v({S^\ast})=v(T)+v({S^\ast\setminus T})$, we have $v({S^\ast}) \le  2v({S_1^\ast})$.  Hence  
\begin{equation*}
v(S) \ge v(S_1)\ge  \rho \cdot v(S_1^\ast)\ge\frac{\rho}{2} v({S^\ast}).
\end{equation*}

Combining Cases (1)-(3): $\min\{\rho, \rho/(1+\rho), \rho/2 \} = \rho/2$, we complete the proof of the approximation ratio of ${\sf Alg}^{\rm a}$ as $\rho/2$.
\end{proof}

Finally, we prove Lemma \ref{lem:subsetsumA}:

\begin{proof} 
The case $\sum_{i=1}^{n} a_i\leq C'$ is trivial.  Otherwise, 
let $j$ be the smallest index such that the partial sum exceeds $C'$, i.e., $\sum_{i=1}^{j-1} a_i \le  C'$ and $\sum_{i=1}^{j} a_i>C'$.  Clearly $j \ge 2$ since all $a_i \le  C'$.

Let $x=\sum_{i=1}^{j-1} a_i$, $z=a_j$ and $y = \sum_{i=j+1}^{n} a_i$.  

Note that $\sum_{i = 1}^{n} a_i = x+y+z$.  We already have 
\begin{equation*}
x \le  C',\ \ z \le  C',\ \ 
x+y+z > C' \mbox{\ \ and\ \ } x+z > C'
\end{equation*}

The lemma holds if $y+z \le  C'$, because we can set $T = \{1 ,..., j-1 \}$.  

If $y+z> C'$, then we obtain:
\begin{eqnarray*}
x+y & = & 2(x+y+z)-(x+z)-(y+z)  \\
& < & 2C-2C'\le  (2-\sqrt{2})C  < \frac{C}{\sqrt{2}}\le C'
\end{eqnarray*}
because $x+y+z \le  C$.
Hence, we can set $T = \{1 ,..., j-1, j+1, ..., n\}$.  
\end{proof}

\section{Monotone Approximation Algorithm for C-KP}\label{sec:monalg}
\noindent
As mentioned in Subsection \ref{subsec:ic}, a monotone polynomial time algorithm for C-KP
implies an incentive compatible polynomial time mechanism.  However, our approximation algorithm ${\sf Alg}^{\rm a}$ presented in Section \ref{sec:algCKP} does not seem to have an easy proof for monotonicity.
In this section, we give a slight modification of ${\sf Alg}^{\rm a}$, for which monotonicity becomes immediate and the approximation ratio is preserved.

\subsection{Basic Idea} \label{subsec:pic}
\noindent
In ${\sf Alg}^{\rm a}$, monotonicity is not guaranteed due to the comparison between $v(S_1)$ and $v(S_2)$, the total value of solution $S_1$ and $S_2$.  Although we assume ${\sf Alg}^{\rm 1d}$ for {\sc 1-KP} is monotone, $v(S_1)$ can fluctuate since $S_1$ is an approximate solution.  Our trick here is to transform each solution candidate for $S_2$, a single item $k$ with demand $d_k\in {\cal D}_2$, to be a solution candidate for $S_1$: an item of the same value whose demand is exactly the capacity limit $C/\sqrt{2}$ for ${\sf Alg}^{\rm 1d}$.  These new items will not combine with each other or with any original items to form new solution candidates for $S_1$.  Then our new algorithm ${\sf Alg}^{\rm b}$ only needs to run ${\sf Alg}^{\rm 1d}$ on the modified set of items to produce a solution for C-KP.

\subsection{Approximation Algorithm} 

\begin{algorithm}[htb!]
\caption{${\sf Alg}^{\rm b} [(d_k,v_k: k \in K), C]$}
\begin{algorithmic}[1]
\FOR{$k \in K$} 
\STATE Set $\tilde{d}_k =  \min\{ \frac{d_k^{\rm R}+d_k^{\rm I}}{\sqrt{2}}, \frac{C}{\sqrt{2}} \}$
\ENDFOR
\STATE Set $S = {\sf Alg}^{\rm 1d}[(\tilde{d}_k,v_k: k \in K),\frac{C}{\sqrt{2}}]$
\STATE Output $S$
\end{algorithmic}
\end{algorithm}

\noindent Recall that we assume every demand $d_k$ lies in ${\cal D}$ ($|d_k|\leq C$).  The preprocessing $\tilde{d}_k =  \min\{ \frac{d_k^{\rm R}+d_k^{\rm I}}{\sqrt{2}}, \frac{C}{\sqrt{2}} \}$ does exactly the transformation mentioned above: For $d_k\in {\cal D}_1$, we simply do the projection onto the $\frac{\pi}{4}$ line; otherwise, $d_k\in {\cal D}_2$, its projection is larger than $C/\sqrt{2}$, and we cut it off to $C/\sqrt{2}$.  Then we run ${\sf Alg}^{\rm 1d}$ on the modified projected demands and outputs the answer.  

The following theorem states that our modification of the algorithm does not change the approximation ratio:
\begin{theorem} 
\label{thm:2apxb}
If ${\sf Alg}^{\rm 1d}$ is a $\rho$-approximation algorithm for {\sc 1-KP}, then ${\sf Alg}^{\rm b}$ is a $\frac{\rho}{2}$-approximation algorithm for {\sc C-KP}.
\end{theorem}

The proof of Theorem \ref{thm:2apxb} is essentially the same as that of Theorem \ref{thm:2apx}.  The main difference is that, now instead of an explicit comparison between the solutions $S_1$ and $S_2$ to the two subproblems on region ${\cal D}_1$ and ${\cal D}_2$ respectively, our algorithm make it implicit inside the execution of ${\sf Alg}^{\rm 1d}$.  Therefore, in the formal proof below, we have to define the two subproblems explicitly and show that the total value of our output $v(S) \ge \rho \cdot \max\{ v(S_1^\ast), v({S_2^\ast})\}$.  

The case analysis is easy given this inequality.  Although ${\sf Alg}^{\rm a}$ has a better approximation guarantee in terms of the inequality $v(S)\geq \max\{ \rho \cdot v(S_1^\ast), v({S_2^\ast})\}$, overall, we achieve the same approximation ratio of $\rho/2$.  Just for case (2), we can only prove an approximation ratio of $\rho/2$, instead of $\rho/(1+\rho)$ for ${\sf Alg}^{\rm a}$.  

\begin{proof}
We partition $K$ into two disjoint sets $K_1$ and $K_2$, such that $K_1 \triangleq \{ k \in K : d_k \in {\cal D}_1 \}$ and $K_2 \triangleq \{ k \in K : d_k \in {\cal D}_2 \}$. Note that the projection of any demand in $K_1$ onto the $\frac{\pi}{4}$ line is at most $C/\sqrt{2}$, whereas that in $K_2$ is larger than $C/\sqrt{2}$.

Let $S_1$ be the output of ${\sf Alg}^{\rm b}$, when the input is $K_1$. Let $S_2$ be the output of ${\sf Alg}^{\rm b}$, when the input is $K_2$.  Let $S_1^\ast$ and $S_2^\ast$ be their corresponding optimal solutions.  $S_1^\ast$ is an optimal solution to {\sc 1-KP} on projected demands within capacity $C/\sqrt{2}$, hence is an optimal solution to {\sc C-KP} on feasible region ${\cal D}_1$.  On the other hand, since each demand in $K_2$ is changed to one exactly equal to the capacity limit of {\sc 1-KP}, only one of them can be satisfied.  Hence $S_2^\ast$ chooses the one with maximum value $S_2^*=\{\arg\max_{k\in K_2} v_k\}$.  

Since any demand in $K_2$ will not combine with any in $K_1$ to form new feasible solutions to {\sc 1-KP}, ${\sf Alg}^{\rm b}$ outputs either a solution whose sum vector lies in ${\cal D}_1$ or a singleton set of a demand in $K_2$, which is evidently a feasible solution to {\sc C-KP}.

Optimally ${\sf Alg}^{\rm 1d}$ would output $\arg\max \{v(S_1^\ast),v(S_2^\ast)\}$.  Since ${\sf Alg}^{\rm 1d}$ is a $\rho$-approximation algorithm to {\sc 1-KP}, we have $v(S) \ge \rho \cdot \max\{ v(S_1^\ast), v({S_2^\ast})\}$. 

Based on this inequality, it is easy to go through the case analysis in the proof of Theorem \ref{thm:2apx} (with slight modifications), hence we omit the rest of the proof here.  
\end{proof}
On the other hand, our new algorithm is monotone according to Definition \ref{def:mon}. 

\begin{theorem} \label{thm:mon}
If ${\sf Alg}^{\rm 1d}$ is a monotone algorithm for {\sc 1-KP}, then ${\sf Alg}^{\rm b}$ is a monotone algorithm for {\sc C-KP}.
\end{theorem}

\begin{proof}
We need to show that, if item $k$ is selected by ${\sf Alg}^{\rm b}$ with demand $d_k$ and value $v_k$, $k$ is also selected with demand $d'_k$ and value $v'_k$, where $v'_k \ge v_k$ and ${d'}_k\leq d_k$ (i.e., ${d'}_k^{\rm R} \le d_k^{\rm R}$ and ${d'}_k^{\rm I} \le d_k^{\rm I}$), while all inputs of other agents do not change.  

Item $k$ is selected by ${\sf Alg}^{\rm b}$ on $d'_k$ and $v'_k$ if and only if it is selected by ${\sf Alg}^{\rm 1d}$ on $\tilde{d}_k'$ and $v'_k$. 
Since $\tilde{d}_k =  \min\{ \frac{{d}_k^{\rm R}+{d}_k^{\rm I}}{\sqrt{2}}, \frac{C}{\sqrt{2}} \}$, ${d'}_k^{\rm R} \le d_k^{\rm R}$ and ${d'}_k^{\rm I} \le d_k^{\rm I}$ implies $\tilde{d}_k' \le \tilde{d}_k$. Then from the monotonicity of ${\sf Alg}^{\rm 1d}$, $k$ is selected by ${\sf Alg}^{\rm 1d}$, and hence by ${\sf Alg}^{\rm b}$.
\end{proof}

Combining Theorem \ref{thm:2apxb}, Theorem \ref{thm:mon} with Theorem \ref{thm:IC} gives:\footnote{Note that algorithms are exact under this problem formulation where a solution is specified as the selection of a subset of items.} 
\begin{corollary}
\label{cor:2apxb}
Since {\sc 1-KP} has a monotone FPTAS \cite{BKV05KS}, there is an incentive compatible $(\frac{1}{2}-\epsilon)$-approximation algorithm for {\sc C-KP} that runs in polynomial-time in the size of input and $1/\epsilon$, for any $\epsilon>0$.
\end{corollary} 

\section{Inapproximability for C-KP}\label{sec:inappr}
\noindent
In this section, we complete the study of {\sc C-KP} by providing an inapproximability result. We show that {\sc C-KP} does not admit an FPTAS, unless P = NP.

We remark that it is known there is no FPTAS for {\sc 2-KP} (see \cite{KPP10book}), which does not have direct implications for {\sc C-KP}.  However, our proof is an extension of the basic idea in the proof for {\sc 2-KP}. 

As in the reduction for 2-KP, we reduce the {\sc equipartition} problem to {\sc C-KP}:

\begin{definition} ({\sc equipartition} Problem): Given a set of positive integers $\{w_k: k \in K \}$, with $|K|=n$ where $n$ is even, we determine if there is a subset of items $S \subseteq K$ such that
\[
|S| = \frac{n}{2} \mbox{\ and\ } \sum_{k \in S} w_k = \sum_{k \notin S} w_k
\]
\end{definition}
It is well-known that {\sc equipartition} is NP-complete.

\begin{theorem}
There is no FPTAS for {\sc C-KP}, unless P = NP.
\end{theorem}
\begin{proof}
We define a decision version of {\sc C-KP} with a cardinality objective: given $\{w_k: k \in K \}$, a capacity bound $C$ and a cardinality bound $M$, we determine if there is a subset of items $S$ such that
\[
|S| \geq M, \mbox{\ and\ } \Big| \sum_{k\in S} d_k \Big| \le C
\]

Now we map every instance of {\sc equipartition} to an instance of the {\sc C-KP} decision problem that always yields the same answer.

Given $\{w_k: k \in K \}$ from {\sc equipartition}, define 
\[
M=n/2,\quad d_k^{\rm R} = w_k, \quad d_k^{\rm I} = \beta(w_{\max} -  w_k), 
\]
\[
C = \sqrt{ \Big( \frac{W}{2} \Big)^2 + \beta^2\Big( \frac{n w_{\max}}{2} -  \frac{W}{2} \Big)^2 }
\]
where $W \triangleq \sum_{k = 1}^{n} w_k$, $w_{\max} \triangleq \max \{ w_k: k \in K\}$.  Note that in our reduction, $d_k^{\rm I} \ge 0$.

As shown in Fig.~\ref{fig:inapprox}, the feasible region $\cal D$ for {\sc C-KP} is the $\frac{1}{4}$ disk of radius $C$ in the first quadrant.  Since for any subset $S\subseteq K$, 
$\sum_{k\in S} d_k^{\rm I}=\beta(|S|\cdot w_{\max}- \sum_{k\in S} d_k^{\rm R})$,
the cardinality constraint $|S|\geq \frac{n}{2}$ imposes all solutions to have its sum vector in the halfplane $H:$ $d^I\geq \beta(\frac{n w_{\max}}{2}- d^R)$.  The dividing line of $H$ goes through point $P:$ $\Big(\frac{W}{2} , \beta(\frac{n w_{\max}}{2} -  \frac{W}{2})\Big)$.  Our main idea is to set $\beta>0$ such that the dividing line of $H$ coincides with the tangent line at $P$.  Thus we make the intersection of $H$ and ${\cal D}$ exactly $P$, which implies    
$|S|=\frac{n}{2}$ and $\sum_{k \in S} w_j=\frac{W}{2}$ for any solution $S$ to our reduced {\sc C-KP} decision problem instance.

\begin{figure}[htb!]
 \centering
 \includegraphics[scale=0.5]{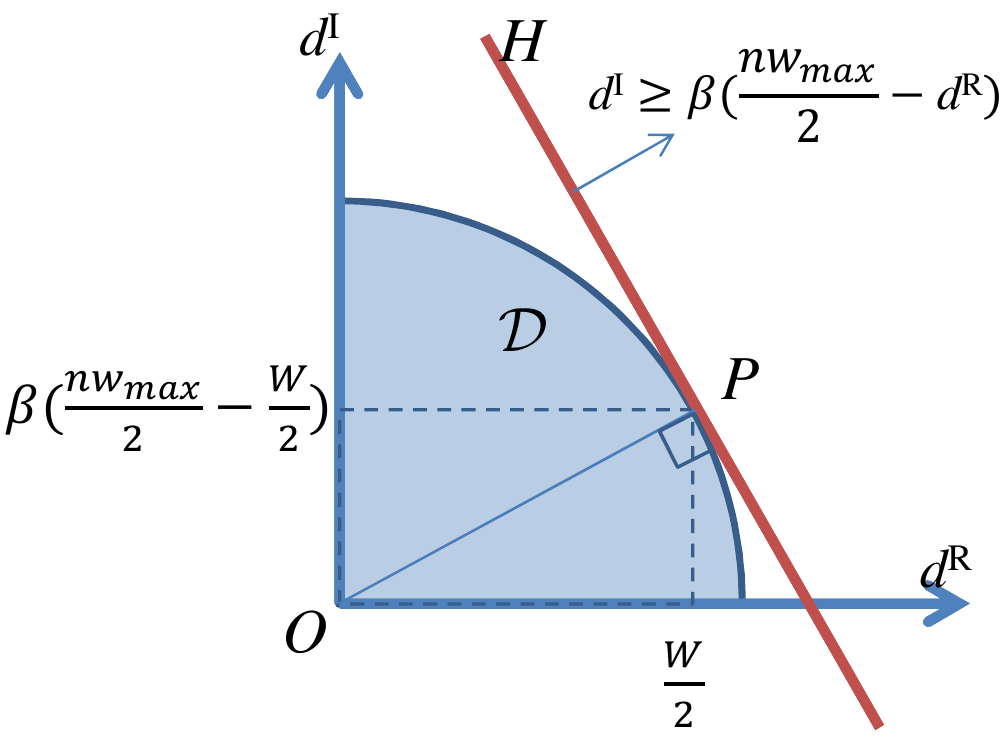} 
 \caption{Reduction of inapproximability.} 
\label{fig:inapprox}
\end{figure}

On the other hand, it is clear that each subset $S$ satisfying conditions of {\sc equipartition} also satisfies conditions of the reduced {\sc C-KP} decision problem.  Therefore, the solution of the reduced {\sc C-KP} decision problem is equivalent to the solution of {\sc equipartition}.

To determine a proper $\beta$, since the dividing line of halfplane $H$ goes through $P$, it coincides with the tangent line at $P$ if and only if they have the same slope, i.e.,
\[
-\frac{\frac{W}{2}}{\beta(\frac{n w_{\max}}{2} - \frac{W}{2})} 
= -{\beta}.
\]
Solving the above equation, we obtain
\[\beta = \sqrt{\frac{W}{n w_{\max}-W}},
\]
which is $>0$ unless all weights are equal.  In this case, we set $\beta=0$, and it is trivially a "yes" instance for both {\sc equipartition} and our C-KP decision problem.

So far we have shown the NP-hardness of the C-KP decision problem.  So its maximization version, where $|S|\geq M$ is replaced by $\max |S|$, is NP-hard.  We use the standard technique to prove the inapproximability of the maximization version by FPTAS. Suppose that there exists an FPTAS for any $\epsilon > 0$ in time polynomial in $n$ and $1/\epsilon$. Then we choose $\epsilon = \frac{1}{n+1}$. Let the optimal solution be $z^\ast>0$ and that of the approximation solution produced by FPTAS be $z^A$. We obtain:
\[
z^A \ge (1-\epsilon) z^\ast > z^\ast - z^\ast/n \ge z^\ast -1
\]
because $z^\ast \le n$. Moreover, since $z^\ast$ is an integer, this implies that the FPTAS can solve the problem exactly in polynomial time, contradicting the NP-hardness of the problem. 

Finally, since the maximization version of C-KP decision problem is a special case of the original C-KP with all $v_k=1$, there is no FPTAS to {\sc C-KP}. 
\end{proof} 
\section{Approximation Algorithm for \\GC-KP}\label{sec:algGCKP}
\noindent
We are also able to solve the generalized problem {\sc GC-KP}, by changing our approximation algorithm ${\sf Alg}^{\rm a}$ in Section \ref{sec:algCKP}. Now, instead of an approximation algorithm ${\sf Alg}^{\rm 1d}$ for {\sc 1-KP} as a subroutine, we rely on an approximation algorithm for {\sc 3-KP} (three-dimensional knapsack problem) as a subroutine. 

\subsection{Basic Idea} \label{subsec:picGen}
\noindent
Now a feasible solution of our problem is a subset of items whose sum of demands lies in the intersection of halfplanes $d^{\rm R}\leq C^{\rm R}$, $d^{\rm I}\leq C^{\rm I}$, and the $1/4$ disk of radius $C$ in the first quadrant.  In the most general case ($C^{\rm R}, C^{\rm I}<C$), both halfplanes cut the circle, which also cut the original regions ${\cal D}_1$ and ${\cal D}_2$ defined in Section \ref{sec:algCKP}.
Fig.~\ref{fig:fig3} shows the new ${\cal D}_1$ (polygon $PSTQO$) and ${\cal D}_2$.  Clearly, the feasible region ${\cal D}$ is the disjoint union of ${\cal D}_1$ and ${\cal D}_2$.  

\begin{figure}[htb!]
 \centering 
 \includegraphics[scale=0.5]{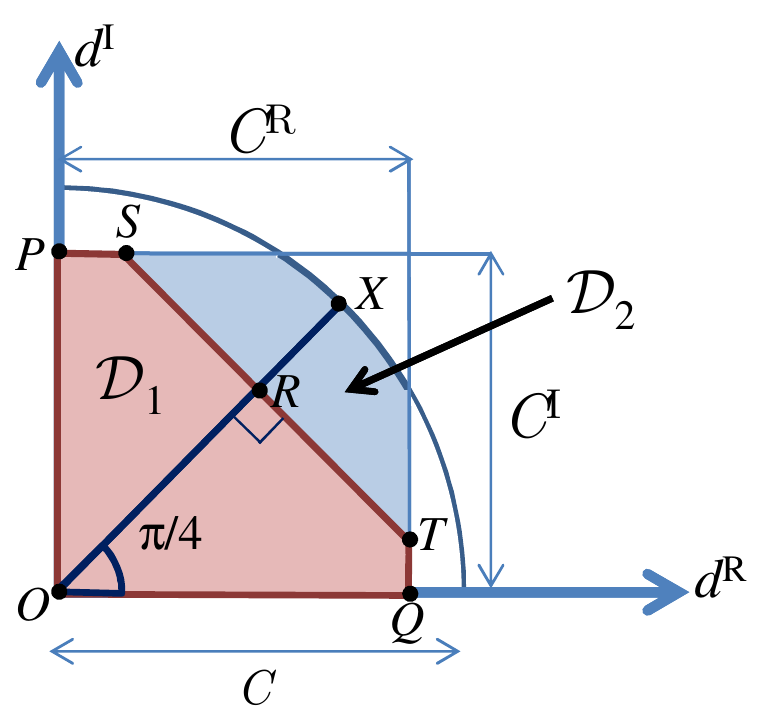} 
 \caption{Graphical Picture for GC-KP.} 
\label{fig:fig3}
\end{figure}

Recall that $OR$ is perpendicular to $ST$ and the length of $OR$ is $C/\sqrt{2}$.  If we denote the projection of a demand $d_k$ onto line $OR$ by $\tilde{d}_k$, the region ${\cal D}_1$ corresponds to 3-dimensional linear constraint $\sum_{k\in K} x_k\tilde{d}_k\leq C/\sqrt{2}$,
$\sum_{k\in K}x_kd_k^{\rm R}\leq C^{\rm R}$ and $\sum_{k\in K}x_kd_k^{\rm I}\leq C^{\rm I}$.  Thus the subproblem on feasible region ${\cal D}_1$ can be solved by a 3-dimensional knapsack algorithm.   

On the other hand, the solutions in polygon ${\cal D}_1$ is almost the whole story by the same reason as in Section \ref{sec:algCKP}.  Again our algorithm takes the maximum between an optimal solution for the subproblem on feasible region ${\cal D}_1$ and an optimal solution on input demands lying in ${\cal D}_2$. 
This reduces the approximation ratio by at most a factor of 2.

The degenerate cases ($C^{\rm R}\geq C$ or $C^{\rm I}\geq C$ or both) can be treated easily by setting $T,Q$ to be the intersection point of the circle and the $d^{\rm R}$-axis, or setting $P,S$ to be the intersection point of the circle and the $d^{\rm I}$-axis, or both.  

\subsection{Approximation Algorithm}
\noindent
Let ${\sf Alg}^{\rm c} [(d_k,v_k: k \in K), C,C^{\rm R},C^{\rm I}]$ be our approximation algorithm for {\sc GC-KP}, where $C,C^{\rm R},C^{\rm I}$ are the capacity on magnitude, real part and imaginary part of total satisfiable demand respectively. Let ${\sf Alg}^{\rm 3d}[((d^1_k, d^2_k, d^3_k),v_k: k \in K), C^1, C^2, C^3]$ be an approximation algorithm for {\sc 3-KP} (e.g., from \cite{FC84alg} or \cite{KPP10book}).   
We describe our approximation algorithm to {\sc GC-KP} as follows:

\begin{algorithm}[htb!]
\caption{${\sf Alg}^{\rm c} [(d_k,v_k: k \in K), C, C^{\rm R}, C^{\rm I}]$}
\begin{algorithmic}[1]
\FOR{$k \in K$} 
\STATE Set $\tilde{d}_k = \frac{d_k^{\rm R}+d_k^{\rm I}}{\sqrt{2}}$
\ENDFOR
\STATE Set $S_1 = {\sf Alg}^{\rm 3d}[((\tilde{d}_k, d_k^{\rm R}, d_k^{\rm I}),v_k: k \in K),\frac{C}{\sqrt{2}}, C^{\rm R}, C^{\rm I}]$
\STATE Set $S_2 = \{\displaystyle {\arg\max}_{k \in K : d_k \in {\cal D}_2  } v_k \}$
\STATE Set $S = \arg\max_{S_1, S_2}\{ v(S_1), v(S_2) \}$
\STATE Output $S$
\end{algorithmic} 
\end{algorithm} 

Our ${\sf Alg}^{\rm c}$ follows the same structure as ${\sf Alg}^{\rm a}$ for C-KP.  The difference is that, for GC-KP, the subproblem on feasible region ${\cal D}_1$ is equivalent to an instance of 3-KP, since ${\cal D}_1$ is defined by three halfplanes.  And to check if a demand $d_k$ lies in region ${\cal D}_2$, we need to check four inequalities: $|d_k|\leq C$, $\tilde{d}_k> C/\sqrt{2}$, $d_k^{\rm R}\leq C^{\rm R}$ and $d_k^{\rm I}\leq C^{\rm I}$. 

\begin{theorem} 
\label{thm:g2apx}
If ${\sf Alg}^{\rm 3d}$ is a $\rho$-approximation algorithm for {\sc 3-KP},  ${\sf Alg}^{\rm c}$ is a $\frac{\rho}{2}$-approximation algorithm for {\sc GC-KP}.
\end{theorem}

\begin{corollary}
\label{cor:g2apx}
Since {\sc 3-KP} has a PTAS \cite{FC84alg}, there is a $(\frac{1}{2}-\epsilon)$-approximation algorithm for {\sc GC-KP} that runs in polynomial-time in the size of input, for any $\epsilon>0$.
\end{corollary}

We omit the proof of Theorem \ref{thm:g2apx} here since it is essentially the same as that of Theorem \ref{thm:2apx} for C-KP.  

\section{Conclusions and Future Work}\label{sec:conc}
\noindent
The knapsack problem has been one of the most popular algorithmic problems since it is a simple abstraction that captures the tradeoff between limited resource and value maximization in resource allocation.  
In this paper, motivated by the need to model AC electrical systems, where power demands have to be represented as complex numbers, we initiate the study of a new variation called the {\em complex-demand knapsack problem} (C-KP).  

By investigating its relationship with multi-dimensional knapsack problems ($m$-KP), we provide $(\frac{1}{2}-\epsilon)$-approximation algorithms for C-KP and its generalization GC-KP; on the other hand, we also show its inapproximability by FPTAS unless P = NP.  Furthermore, our approximation algorithm for C-KP can be monotonized, which implies the existence of a mechanism/algorithm of the same approximation ratio that is incentive compatible, individually rational, and computationally efficient.                

Our results provide basic insights on efficient power allocation in AC electrical systems, which is a fundamental problem in the design of multi-agent systems for smart grid.
Still, there are interesting directions to continue in the future: First, we hope to find a PTAS for C-KP, closing the gap between constant approximation and FPTAS, and a monotone algorithm for GC-KP.  Second, we will extend the problem to an electrical network setting, where there is an underlying network connecting different devices with links of different capacities on the magnitude of transmitted power.  For mechanism design, we may require an additional property called {\em cancellability}: the total payment to be collected from the agents should always cover the cost to generate the power supply, given a cost function of power generation.  We are not aware of any previous work related to this property in mechanism design, and we expect new insights and techniques coming out of the study on it.       

\medskip
\noindent {\bf Acknowledgments\ } 
We thank Mario Szegedy and anonymous referees for useful suggestions.  Lan Yu is supported by Singapore NRF Research Fellowship 2009-08.  Chi-Kin Chau is supported by MI-MIT Collaborative Research Project (11CAMA1).
\bibliographystyle{abbrv}
\bibliography{reference}



\end{document}